\documentclass{llncs}

\usepackage{amssymb}
\usepackage{amsmath}
\usepackage{textcomp}
\usepackage{array}
\usepackage[latin5]{inputenc}
\usepackage{lineno}
\usepackage{color}
\usepackage{comment}

\usepackage{graphicx}

\newcommand{\ket}[1]{|#1\rangle}
\newcommand{\braket}[2]{\langle #1|#2\rangle}
\newcommand{\cent}[0]{\mbox{\textcent}}
\newcommand{\dollar}[0]{\$}

\newtheorem{fact}{Fact}
\newcommand{\tinyspace}{\mspace{1mu}}
\newcommand{\abs}[1]{\left\lvert\tinyspace #1 \tinyspace\right\rvert}
\newcommand{\inner}[2]{\ensuremath{\langle{#1}|{#2}\rangle}}
\title{Implications of quantum automata for contextuality}

\author{Jibran Rashid\inst{1}$^,$\thanks{Jibran Rashid was supported by QSIT Director's Reserve Project.} \and Abuzer Yakary{\i}lmaz\inst{2}$^,$\thanks{Abuzer Yakary{\i}lmaz was partially supported by ERC Advanced Grant MQC.}}

\institute{Facolt\`{a} di Informatica, Universit\`{a} della Svizzera Italiana, Via G. Buffi 13, 6900, Lugano, Switzerland
\and
University of Latvia, Faculty of Computing, Raina bulv. 19, R\={\i}ga, LV-1586, Latvia \\
\email{jibran.rashid@usi.ch}, \email{abuzer@lu.lv}
}

\authorrunning{Rashid and Yakary{\i}lmaz} % abbreviated author list (for running head)

\begin{document}

\maketitle

\begin{abstract}
We construct zero-error quantum finite automata (QFAs) for promise problems which cannot be solved by bounded-error probabilistic finite automata (PFAs).  Here is a summary of our results:
\begin{enumerate}
	\item There is a promise problem solvable by an exact two-way QFA in exponential expected time, but not by any bounded-error sublogarithmic space probabilistic Turing machine (PTM).
	\item There is a promise problem solvable by an exact two-way QFA in quadratic expected time, but not by any bounded-error $ o(\log \log n) $-space PTMs in polynomial expected time. The same problem can be solvable by a one-way Las Vegas (or exact two-way) QFA with quantum head in linear (expected) time.
	\item There is a promise problem solvable by a Las Vegas realtime QFA, but not by any bounded-error realtime PFA. The same problem can be solvable by an exact two-way QFA in linear expected time but not by any exact two-way PFA.
	\item There is a family of promise problems such that each promise problem can be solvable by a two-state exact realtime QFAs, but, there is no such bound on the number of states of realtime bounded-error PFAs solving the members this family.
\end{enumerate}
Our results imply that there exist zero-error quantum computational devices with a \emph{single qubit} of memory that cannot be simulated by any finite memory classical computational model. This provides a computational perspective on results regarding ontological theories of quantum mechanics~\cite{Hardy04},~\cite{Montina08}. As a consequence we find that classical automata based simulation models \cite{Kleinmann11},~\cite{Blasiak13} are not sufficiently powerful to simulate quantum contextuality. We~conclude by highlighting the interplay between results from automata models and their application to developing a general framework for quantum contextuality. 
\end{abstract}
%
% NEW COMMANDS
%

\newcommand{\pal}{\mathtt{PAL}}
\newcommand{\upal}{\mathtt{UPAL}}
\newcommand{\eq}{\mathtt{EQ}}
\newcommand{\twin}{\mathtt{TWIN}}

\newcommand{\promisepal}{\mathtt{PromisePAL}}
\newcommand{\promisepalyes}{\mathtt{PromisePAL_{yes}}}
\newcommand{\promisepalno}{\mathtt{PromisePAL_{no}}}
\newcommand{\evenodd}{\mathtt{EVENODD^k}}
\newcommand{\promisetwinpal}{\mathtt{PromiseTWINPAL}}
\newcommand{\promisetwinpalyes}{\mathtt{PromiseTWINPAL_{yes}}}
\newcommand{\promisetwinpalno}{\mathtt{PromiseTWINPAL_{no}}}

\newcommand{\promiseexptwinpal}{\mathtt{EXPPromiseTWINPAL}}
\newcommand{\promiseexptwinpalyes}{\mathtt{EXPPromiseTWINPAL_{yes}}}
\newcommand{\promiseexptwinpalno}{\mathtt{EXPPromiseTWINPAL_{no}}}

\newcommand{\promiseeq}{\mathtt{PromiseEQ}}
\newcommand{\promiseeqyes}{\mathtt{PromiseEQ_{yes}}}
\newcommand{\promiseeqno}{\mathtt{PromiseEQ_{no}}}

\newcommand{\exactpal}{\mathcal{EXACT_{PAL}}}
\newcommand{\exacteq}{\mathcal{EXACT_{EQ}}}
\newcommand{\exacttwinpal}{\mathcal{EXACT_{TWINPAL}}}
\newcommand{\awpal}{\mathcal{AW_{PAL}}}
\newcommand{\aweq}{\mathcal{AW_{EQ}}}
\newcommand{\awtwinpal}{\mathcal{AW_{TWINPAL}}}	
\newcommand{\lvexppal}{\mathcal{LV_{EXPTWINPAL}}}	
\newcommand{\exactexppal}{\mathcal{EXACT_{EXPTWINPAL}}}	
\pretolerance=10000

\newcommand{\op}[1]{\textsf{#1}}
\newcommand{\expect}[1]{\ensuremath{\langle{#1}\rangle}}

%\section{Introduction}

%Promise problem is a generalization of language ` exactly but not by any deterministic pushdown automaton. Recently, Ambainis and Yakary{\i}lmaz \cite{AY12} showed that there is an infinite family of promise problems which can be solved exactly by just tuning transition amplitudes of a realtime two-state quantum finite automaton (QFA), whereas the size of the corresponding classical automata grows without bound. Two more recent papers regarding succinctness QFAs of are \cite{ZQGLM13} and \cite{GQZ13}.

%In this paper, we present several new results regarding QFAs without making any error on some promise problems which cannot be solvable by bounded-error probabilistic finite automata (PFAs) or bounded-error probabilistic Turing machines (PTMs) with small space bounds.

\section{Preliminaries}

Consider Alice and Bob who are presented with a $3 \times 3$ grid as depicted in Figure~\ref{fig:square}. They are asked to determine entries $\op{A}_i \in \{-1,+1\}$ for each cell in the grid such that the parity for each row and column is ``$+1$'' except for the third column which has parity ``$-1$''. Let $\op{R}i$ be the parity for row~$i$ and $\op{C}j$ be the parity for column~$j$. The fact that no such assignment exists for the square can be verified by noting that $\prod_{i=1}^3 \op{R}i=1$ while $\prod_{j=1}^3 \op{C}i=-1$.
\begin{figure}[ht]
\centering
\scalebox{1}{\includegraphics{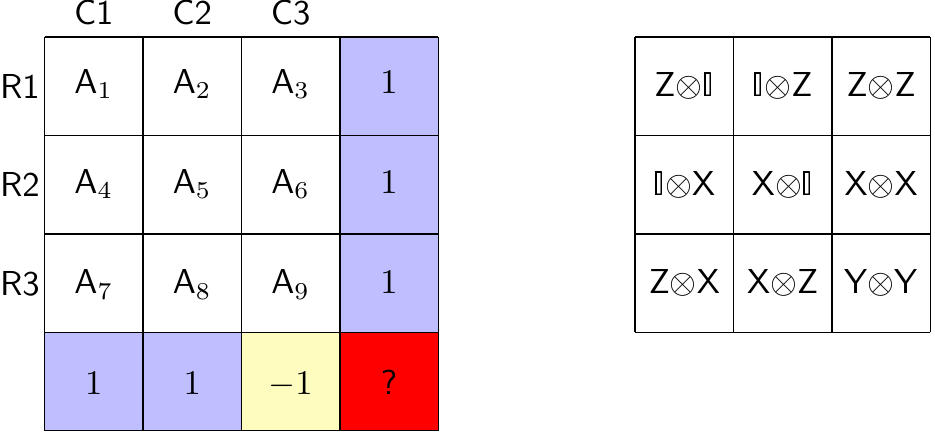}} 
\caption{Peres--Mermin magic square on the left. Each entry in the right square gives the measurement performed by the players to generate the corresponding output bit for the Peres--Mermin square.}
\label{fig:square}
\end{figure}

After determining a common strategy, the players are spatially separated and the game proceeds as follows. Alice receives input $i$ and Bob receives input $j$, each chosen uniformly random from the set~$\{1,2,3\}$. They are required to output cell entries corresponding to row~$i$ and column~$j$ respectively, such that the parity requirement is satisfied and furthermore the common cell in their output is consistent, i.e., both of them assign it the same value. The game is called the Peres-Mermin magic square~\cite{Peres90},~\cite{Merm90a} and is an example of the more general Kochen--Specker theorem~\cite{Kochen67}.

Even though no classical strategy allows the players to win the magic square game with certainty, if the players share a pair of Bell states given by
\begin{equation*}
\ket{\psi} = \frac{1}{2}\left(\ket{00}+\ket{11} \right)^{\otimes 2},  
\end{equation*}
then performing the measurements given in Figure~\ref{fig:square} result in correlations that always satisfy the magic square requirements. This does not correspond to a fixed assignment to the square, just that in each independent run of the game, Alice and Bob are able to generate output that satisfies the requirements imposed on the rows and columns. This behaviour is usually studied under the notion of \emph{contextuality} in quantum mechanics.

Recent work has focused on developing a general framework for contextuality based on generating a hypergraph for a given contextuality scenario and studying its combinatorial properties~\cite{Cabello14},\cite{Acin14}. Even though graph theoretic structures are appropriate for modelling contextuality they lack the computational perspective that comes from modelling the computational procedures that generate contextuality scenarios. Quantum automata provide exactly such a framework. As a direct consequence of such considerations, we find that separations between classical and quantum finite automata imply that no amount of finite memory is in general sufficient to simulate quantum behaviour. Similar results have also been obtained by Hardy~\cite{Hardy04} and Montina~\cite{Montina08}. 

Kleinmann et al.~\cite{Kleinmann11} and Blasiak~\cite{Blasiak13} have suggested a classical simulation of Peres-Mermin magic square using classical memory. Cabello and Joosten~\cite{Cabello11} have shown that the amount of memory required to simulate the measurement results of the generalized Peres-Mermin square increasingly violate the Holevo bound. Cabello~\cite{Cabello12b} proposed the \emph{principle of bounded memory} which states that the memory a finite physical system can keep is bounded. On the other hand, Cabello~\cite{Cabello12} has also shown that the memory required to produce quantum predictions grows at least exponentially with the number of qubits~$n$. 

We~show that a stronger statement follows from our results on the separations between quantum and classical finite automata. More specifically, there exist promise problems that quantum automata equipped with a \emph{single qubit} can solve with zero-error while no classical finite memory model can solve these problems with bounded error. In contrast, the exponential separation obtained by Cabello~\cite{Cabello12} requires a quantum system of size~$n$. The hidden variable model for a single qubit due to Bell~\cite{Bell64} does not apply since there are only finite bits available for the classical simulation.
 
We assume the reader is familiar with the basics of quantum computation~\cite{NC00} and the basic models in automata theory~\cite{Yak11A}. 

\subsection{Quantum Automata}
We denote input alphabet by $ \Sigma $, which does not include $ \cent $ (the left end-marker), $ \dollar $ (the right end-marker) and $ \tilde{\Sigma} = \Sigma \cup \{\cent,\dollar\} $. A promise problem is a pair $ \tt P = (P_{yes},P_{no}) $, where $ \mathtt{P_{yes},P_{no}} \subseteq \Sigma^* $ and $ \mathtt{P_{yes}} \cap \mathtt{P_{no}} = \emptyset $ \cite{Wat09A}. $ \tt P $ is said to be solved by a machine $ \mathcal{M} $ with error bound $ \epsilon \in (0,\frac{1}{2}) $ if any member of $ \tt P_{yes} $  is accepted with a probability at least $ 1 - \epsilon $ and any member of $ \tt P_{no} $ is rejected by $ \mathcal{M} $ with a probability at least $ 1- \epsilon $. $ \tt P $ is said to be solved by $ \mathcal{M} $ with bounded-error if it is solved by $ \mathcal{M} $ with an error bound. If $ \epsilon = 0 $, then it is said that the problem is solved by $ \mathcal{M} $ \textit{exactly}. A special case of bounded-error is one-sided bounded-error where either all members of $ \tt P_{yes} $ are accepted with probability 1 or all members of $ \tt P_{no} $ are rejected with probability 1. $ \mathcal{M} $ is said to be Las Vegas (with a success probability $ p \in (0,1] $) \cite{HS01} if
\begin{itemize}
	\item $ \mathcal{M} $ has the ability of giving three answers (instead of two): ``accept'', ``reject'', or ``don't know'';
	\item for a member of $\tt P_{yes}$, $ \mathcal{M} $ gives the decision of ``acceptance'' with a probability at least $ p $ and gives the decision of ``don't know'' with the remaining probability; and,
	\item for a member of $ \tt P_{no}$, $ \mathcal{M} $ gives the decision of ``rejection'' with a probability at least $ p $ and gives the decision of ``don't know'' with the remaining probability.
\end{itemize} 
If $ \tt P $ satisfies $ \mathtt{P_{yes}} \cup \mathtt{P_{no}} = \Sigma^* $ and it is \textit{solvable} by $ \mathcal{M} $, then it is conventional said that $ \mathtt{P_{yes}} $ is \textit{recognized} by $ \mathcal{M} $.

All models in this paper have a single-head read-only tape on which the given input string is placed between left and right end-markers. The head never leaves the end-markers. The input head can move to the left, move to the right, or stay on the same square. This property is denoted as ``two-way''. If the input head is not allowed to move to left, then it is called ``one-way''. As a further restriction, if the input head is allowed to stay on the same square only for a fixed-number of steps, then it is called ``realtime''. Note that any realtime quantum model has a classical head.

A two-way automaton is called sweeping if the input head is allowed to change its direction only on the end-markers \cite{Sip80,KKM12}. A very restricted version of sweeping automaton called restarting realtime automaton runs a realtime algorithm in an infinite loop,  \cite{YS10B}, i.e. if the computation is not terminated on the right end-marker, the same realtime algorithm is executed again. 

A two-way finite automaton with quantum and classical states (2QCFA) \cite{AW02} is a two-way deterministic finite automaton augmented with a fixed-size quantum register. Formally,\footnote{Here, we define a slightly different model than the original one, but, they can simulate each other exactly.} a 2QCFA is 
\begin{equation*}
	\mathcal{M} = (S,Q,\Sigma,\delta,s_1,q_1,s_a,s_r),
\end{equation*}
where $ S $ and $Q$ are the set of classical and quantum states, respectively; $ s_1 \in S $ and $q_1 \in Q$ are initial classical and quantum states, respectively; $ s_a \in S $ and $ s_r \in S $ ($s_a \neq s_r$) are the accepting and rejecting states, respectively; and $ \delta $ is the transition function composed by two sub-elements $ \delta_q $ and $ \delta_c $ that govern the quantum part and classical part of the machine, respectively. Suppose that $ \mathcal{M} $ is in state $ s \in S $ and the symbol under the input head is $ \sigma \in \tilde{\Sigma} $.  In each step, first the quantum part and then the classical part is processed in the following manner:
\begin{itemize}
	\item $ \delta_q(s,\sigma) $ determines either a unitary operator, say $ \op{U}_{s,\sigma} $, or a projective operator, say $ \op{P}_{s,\sigma} = \{ \op{P}_{s,\sigma,1},\ldots,\op{P}_{s,\sigma,k} \} $ for some $ k>0 $, and then it is applied to the quantum register. Formally, in the former case,
\begin{equation*} \delta_q(s,\sigma,\ket{\psi}) \rightarrow (i=1,U\ket{\psi}), \end{equation*} and, in the latter case, 
{\small
\begin{equation*} 
	\delta_q(s,\sigma,\ket{\psi}) \rightarrow \left\lbrace \left(i, \frac{\ket{\psi_i}}{\sqrt{\braket{\psi_i}{\psi_i}}} \right) \middle| ~ \ket{\psi_i} = \op{P}_{s,\sigma,i} \ket{\psi}, \braket{\psi_i}{\psi_i} \neq 0,  \mbox{ and } 1 \leq i \leq k. \right\rbrace .
\end{equation*}
}
We fix $i = ``1$'' if a unitary operator is applied. Note that only a single outcome ($ i \in \{1,\ldots,k\}$) can be observed in the case of a projective measurement.
	\item After the quantum phase, the machine evolves classically. Formally, 
\begin{equation*}
	\delta(s,\sigma,i) \rightarrow (s',d),
\end{equation*}
where $ i $ is the measurement outcome of quantum phase, $ s' $ is the new classical state, and $ d \in \{ \leftarrow, \downarrow, \rightarrow \} $ represents the update of the input head.
\end{itemize}
Note that, for Las Vegas algorithms, we need to define another halting state called $ s_d $ corresponding to answer ``don't know''.

The computation of $ \mathcal{M} $ on a given input string $ w $ starts in the initial configuration, where the head is on the first symbol of $ \tilde{w} = \cent w \dollar $, the classical state is $ s_1 $, and the quantum state is $ \ket{q_1} $. The computation is terminated and the input is accepted (resp., rejected) if $ \mathcal{M} $ enters to state $ s_a $ (resp., $s_r$). 

A 2QCFA restricted to a realtime head is denoted by rtQCFA. Formally, defined in \cite{ZQLG12}, on each tape square a rtQCFA applies a unitary operator followed by a projective measurement, and then evolves its classical part.\footnote{This definition is sufficient to obtain the most general realtime quantum finite automaton \cite{Hir10,YS11A}. Moreover, allowing more than one quantum or classical transition on the same tape square does not increase the computational power of rtQCFAs.} Another restricted realtime QFA model is the Moore-Cruthcfield quantum finite automaton (MCQFA) \cite{MC00}. It consists of only quantum states and a single unitary operator determined by the scanned symbol is applied on each tape square. A projective measurement is applied at the end of computation. A probabilistic or quantum automaton is called rational or algebraic if all the transitions are restricted to rational or algebraic numbers.

\section{Quantum automata for promise problems}

% Give listing
In this section, we present some promise problems solvable by QFAs \textit{without error} but not solvable by their \textit{bounded-error} probabilistic counterparts. At the end, we will also show that the family of promise problem, which was shown to be solvable by a family of exact realtime QFAs (MCQFAs) having only two states \cite{AY12}, cannot be solvable by a family of bounded-error probabilistic finite automata (PFAs) having a fixed number of states. % We classify our result by the names of quantum models.

\subsection{Exact rational (sweeping) 2QCFA algorithm}

2QCFAs can recognize $ \pal = \{ w \mid w \in \{a,b\}^* \mbox{ and } w = w^r \} $ 
%and $ \eq = \{ a^nb^n \mid n > 0 \} $ 
for any one-sided error bound \cite{AW02,YS10B}. In the case of one-sided error, one decision is always reliable. We use this fact to develop quantum automata for solving promise problems inherited from $ \pal $. 
%and $ \eq $. 
We~know that $ \pal $ cannot be recognized by bounded-error PTMs using sublogarithmic space \cite{DS92,FK94}.
% and $ \eq $ can be recognized by bounded-error $ o(\log \log n) $-space PTMs only in super-polynomial expected time \cite{GW86,DS90}. 
We take into consideration these facts when formulating our promise problems so that the impossibility results for bounded-error probabilistic algorithms are still applicable for our constructions.

Our first promise problem is 
$
	\promisepal = (\promisepalyes,\promisepalno), 
$
where 
\begin{itemize}
	\item $ \promisepalyes = \{ ucv | u,v \in \{a,b\}^*, |u|=|v|, u \in \pal, \mbox{ and } v \notin \pal \} $ and
	\item $ \promisepalno ~ = \{ ucv | u,v \in \{a,b\}^*, |u|=|v|, u \notin \pal, \mbox{ and } v \in \pal \} $.
\end{itemize}
Each of the two 2QCFA algorithms given for $\pal$ in~\cite{AW02} and~\cite{YS10B} have zero-error when they reject. That is, for a given $ \epsilon \in (0,\frac{1}{2}) $, there exists a 2QCFA $\mathcal{M_{\epsilon}}$ which always accepts every string $w \in \pal$ and every $w \notin \pal$ is accepted with probability at most $ \epsilon $ and it is rejected with probability at least $ 1-\epsilon $. So, if $ \mathcal{M_{\epsilon}} $ rejects an input, we can be certain that the input is a non-member. 

We can design an exact 2QCFA, say $ \exactpal $, for $\promisepal$ based on $ \mathcal{M_{\epsilon}} $ as follows: Let $ w = ucv \in \promisepal $ be the input such that $u,v \in \{a,b\}^*$ and $|u|=|v|$. On input string $w$, $\exactpal $ proceeds in an infinite loop as follows,
\begin{itemize}
	\item the computation splits into two branches on the left end-marker with probabilities  $ \frac{16}{25} $ and $ \frac{9}{25} $, respectively, by applying a rational unitary operator to a qubit followed by a measurement in the computational basis;
	\item in the $1^{st}$ branch, $ \exactpal $ executes $ \mathcal{M_{\epsilon}} $ on $ v $ and \textit{accepts} $ w $ if $ \exactpal $ \textit{rejects} $ v $;
	\item in the $ 2^{nd} $ branch, $ \exactpal $ executes $ \mathcal{M_{\epsilon}} $ on $ u $ and \textit{rejects} $ w $ if $ \exactpal $ \textit{rejects} $ u $; and,
	\item the computation continues, otherwise.
\end{itemize}
Note that only a single decision is given in each branch: In the $ 1^{st} $ branch, the members of $ \promisepalyes $ are accepted with a probability at least $ 1 - \epsilon $ and no decision is given on the members of $ \promisepalno $. In the $ 2^{nd} $ branch, no decision is given on the members of $ \promisepalyes $ and the members of $ \promisepalno $ are rejected with a probability at least $ 1-\epsilon $. Thus, in a single round, the members of $ \promisepalyes $ are accepted with a probability at least $ \frac{16}{25} (1-\epsilon) $ and the members of $ \promisepalno $ are rejected with a probability at least $ \frac{9}{25} (1-\epsilon) $. Thus, $\exactpal$ separates $ \promisepalyes $ and $ \promisepalno $ exactly by calling $ \mathcal{M_{\epsilon}} $ in expected linear time. This establishes Theorem~\ref{thm:promisepal} while the fact that sublogarithmic space PTMs cannot solve~$\promisepal$ is established in Theorem~\ref{thm:nosubPTMforPAL}, proof of which can be found in Appendix~\ref{app:nosubPTMforPAL}.

\begin{theorem}\label{thm:promisepal}
	$\promisepal$ can be solvable by an exact rational sweeping 2QCFA in exponential expected time.
\end{theorem}

\begin{theorem}
	\label{thm:nosubPTMforPAL}
	Bounded-error sublogarithmic space PTMs cannot solve $ \promisepal $.
\end{theorem}

The scheme given above can be easily generalized to many other cases. The size of the quantum register, the type of the head, and the type of the transitions are determined by $ \mathcal{M_\epsilon} $. Specifically, (i) if $ \mathcal{M_\epsilon} $ is restarting (sweeping), then $ \exactpal $ is restarting (sweeping), too, or (ii) if $ \mathcal{M_\epsilon} $ has only rational (algebraic) amplitudes, then  $ \exactpal $ has rational (algebraic) amplitudes. 

The 2QCFA algorithm for $\pal$ given by Ambainis and Watrous \cite{AW02} is rational and sweeping. The one given by Yakary{\i}lmaz and Say in \cite{YS10B} is restarting but uses algebraic numbers. Both of them run in expected exponential time. In the next subsection we present a new promise problem and we use the former algorithm to obtain an exact rational restarting rtQCFA. Currently, we do not know how to obtain a similar result based on the latter model except by utilizing superoperators. 

\subsection{Exact rational restarting rtQCFA algorithm}

Now, we define a promise problem (a modified version of $\promisepal$) solvable by an exact rational restarting rtQCFA but not by any sublogarithmic space PTMs: 
\normalsize{
$
	\promisetwinpal = (\promisetwinpalyes,\promisetwinpalno),
$}
\normalsize{where}
%\small
\begin{itemize}
	\item $ \promisetwinpalyes = \{ u c u c v c v | u,v \in \{a,b\}^+, |u|=|v|, u \in \pal, \mbox{ and } v \notin \pal \} $, and
	\item $ \promisetwinpalno = \{ u c u c v c v | u,v \in \{a,b\}^+, |u|=|v|, u \notin \pal, \mbox{ and } v \in \pal \} $.
\end{itemize}

\normalsize
\begin{theorem}
\label{thm:TWIN}
	There is an exact rational restart rtQCFA that solves $\promisetwinpal$ in exponential expected time. (See Appendix~\ref{app:TWIN} for the proof)
\end{theorem}

\begin{theorem}
	\label{thm:nosubPTMforTWIN}
	$\promisetwinpal$ cannot be solved by any bounded-error $ o(\log n) $-space PTM. (See Appendix~\ref{app:nosubPTMforTWIN} for the proof)
\end{theorem}  

\subsection{Las Vegas rational rtQCFA algorithm}

Here, we present another promise problem solvable by Las Vegas rtQCFAs or linear-time exact 2QCFAs but not by any bounded-error realtime PFA (rtPFA). Since any exact 2PFA can be simulated by a realtime deterministic finite automaton (rtDFA) (Appendix \ref{app:exact2PFAs}), exact two-way PFAs (2PFAs) also cannot solve this new promise problem. 
The new promise problem is given by: 
$
	\promiseexptwinpal = (\promiseexptwinpalyes,\promiseexptwinpalno),
$
where 
\normalsize{
\begin{itemize}
	\item $ \promiseexptwinpalyes = \{ (u c u c v c v c)^{t} | u,v \in \{a,b\}^+, |u|=|v|, u \in \pal, v \notin \pal, \mbox{ and } t \geq 25^{|u|}  \} $, and
	\item $ \promiseexptwinpalno = \{ (u c u c v c v c)^{t} | u,v \in \{a,b\}^+, |u|=|v|, u \notin \pal, v \in \pal, \mbox{ and } t \geq 25^{|u|} \} $.
\end{itemize}
}

\normalsize
\begin{theorem}
	\label{thm:promiseexptwinpal}
	$ \promiseexptwinpal $ can be solved by a Las Vegas rational rtQCFA or by an exact rational restarting rtQCFA in linear expected time. (See Appendix \ref{app:promiseexptwinpal} for the proof.)
\end{theorem}

\begin{theorem}
	\label{thm:nortPFA}
	There is no bounded-error rtPFA that solves $ \promiseexptwinpal $. (See Appendix~\ref{app:nortPFA} for the proof.)
\end{theorem} 

\subsection{Polynomial-expected-time algebraic exact restarting rtQCFA algorithm}
\label{sec:polynomial}

Our promise problem for this section is as follows:
\[
	\promiseeq = (\promiseeqyes,\promiseeqno), \mbox{ where}
\]
\[ 
	\promiseeqyes = \{ a^mba^mba^n | m \neq n \} \mbox{ and } \promiseeqno ~ = \{  a^mba^nba^m | m \neq n \}.
\]

\begin{theorem}
	\label{thm:polynomial}
	There is an exact algebraic restarting rtQCFA algorithm solving $ \promiseeq $ in quadratic expected time. Moreover, $ \promiseeq $ can be solved by a Las Vegas rational one-way QFA in linear time or by an exact rational two-way QFAs in linear expected time, where both models have a quantum head \cite{KW97,Yak12C}. (See Appendix \ref{app:polynomial} for the proof)
\end{theorem}

\begin{theorem}
	\label{thm:promiseeq}
	$ \promiseeq $ cannot be solved by any bounded-error $ o(\log \log n) $-space PTMs in sub-exponential expected time. (See Appendix \ref{app:promiseeq} for the proof)
\end{theorem}

\subsection{Succinctness of realtime QFAs}
\label{sec:succinctness}

For a given positive integer $ k $, $ \tt EVENODD^k = (EVENODD^k_{yes},EVENODD^k_{no}) $ is a promise problem \cite{AY12} such that
\begin{itemize}
\item $ \mathtt{EVENODD^k_{yes}}=\{a^{i2^k} \mid i \mbox{ is a nonnegative even integer}\}$, and 
\item $ \mathtt{EVENODD^k_{no}}~=\{a^{i2^k} \mid i \mbox{ is a nonnegative odd integer}\}$.
\end{itemize}
Ambainis and Yakary{\i}lmaz \cite{AY12} showed that $ \tt EVENODD^k $ can be solved by a 2-state MCQFA exactly, but, the corresponding probabilistic automaton needs at least $ 2^{k+1} $ states. We~show in Theorem~\ref{thm:succintness} that allowing errors in the output does not help in decreasing the space requirement. Proof of Theorem~\ref{thm:succintness} is provided in Appendix~\ref{app:succintness}.

\begin{theorem}
	\label{thm:succintness}
	Bounded-error rtPFAs need at least $ 2^{k+1} $ states to solve $ \tt EVENODD^k $.
	\end{theorem}

\section{Noncontextual inequalities from automata}

We begin by reformulating the Peres-Mermin game in terms of inequalities. Let $\langle \textsf{A}_i \textsf{A}_j \textsf{A}_k \rangle$ be the expected parity of the corresponding entries of the square. We~associate with each strategy, a value of the game $\expect{\chi}$, which is given by
\begin{equation*}
\begin{array}{ll@{\;\langle\;}l@{\;\rangle\;}l@{\;\langle\;}l@{\;\rangle\;}l@{\;\langle\;}l@{\;\rangle}l}
\expect{\;\chi\;} =&\hspace*{-.25em}& \textsf{A}_1\textsf{A}_2\textsf{A}_3 & + & \textsf{A}_4\textsf{A}_5\textsf{A}_6 & + & \textsf{A}_7\textsf{A}_8\textsf{A}_9&\\
&\hspace*{-.25em}+&\textsf{A}_1\textsf{A}_4\textsf{A}_7 & + & \textsf{A}_2\textsf{A}_5\textsf{A}_8 & - & \textsf{A}_3\textsf{A}_6\textsf{A}_9&.
\end{array}
\end{equation*} 
The classical bound is $\expect{\chi} \leqslant 4$, while the quantum bound is given by $\expect{\chi} \leqslant 6$. We~can now construct similar inequalities for the promise problems defined in this paper. The general idea behind the inequalities is to construct a game based on quantum and classical automata separations. Assume Bob is restricted to either~$N$ bits of classical memory or~$N$ quantum bits and Alice has the task of verifying what type of memory is available to Bob. She can query Bob multiple times on a pre-selected problem that is known to both of them. Conditioned on the classical memory requirement for the problem the idea then is for Alice to iteratively query Bob on input strings of increasing length. Eventually Bob's classical memory becomes insufficient to correctly answer the query and his best response is a random guess.

$\pal$ and $\promisepal$ can be solved in log space. As a consequence, Alice requires an exponential number of queries in~$N$ before Bob's memory is exhausted. The classical exponential memory requirement for $\evenodd$ means that number of queries need only be logarithmic before a violation is observed. On the other hand, $\evenodd$ is not a single problem but a family of promise problems and the classical memory requirement is for rtPFAs. For $\promisepal$ we obtain zero error for the quantum strategy while for the classical strategy bounded-error is not possible for 2PFAs.

We base the inequality we present on $\evenodd$. The arguments carry over to $\pal$ and $\promisepal$ as well. 
%Let the sets $ \tt EVENODD^k_{yes}$ and $ \tt EVENODD^k_{no}$ be the sets from the promise problem $\evenodd$ defined in Section~\ref{sec:succinctness}. 
On a given query Bob receives as input an integer~$k$ and a unary string~$w=a^l$ that is promised to be from either $\tt EVENODD^k_{yes}$ or $ \tt EVENODD^k_{no}$, i.e., $l=i2^k$. The task for Bob is to determine the membership of string~$w$, i.e., whether~$i$ is even or odd. So, he outputs~``$+1$'' if~$i$ is even and~``$-1$'' otherwise. The identification can be made for any~$k$ if Bob has unbounded memory. If Bob is restricted to have memory~$2^{n+1}$ then the identification can still be made perfectly for all integer inputs to Bob with~$k \leqslant n$. It becomes impossible to perform this identification perfectly when~$k>n$. In this case the amount of memory available to Bob is not sufficient to determine classically the membership of the input string~$w$.

In the quantum case, a perfect strategy exists for all~$k$, if Bob is allowed access to a single qubit~$\ket{\psi}$. The state is initialized to~$\ket{0}$ and for each~``$a$'' in the string~$w$ Bob applies the rotation
$
\op{U}_a = 
\begin{pmatrix} 
\cos \theta & -\sin \theta \\
\sin \theta & \hphantom{-}\cos \theta
\end{pmatrix},
$
with~$\theta = \frac{\pi}{2^{k+1}}$. Bob measures in the computational basis once the input is processed and realizes that~$i$ is odd if he obtains~$\ket{1}$ or~$-\ket{1}$ and~$i$ is even is the result is~$\ket{0}$ or~$-\ket{0}$. This procedure guarantees that Bob will always correctly identify the string~$w$.

Assume that the amount of memory available to Bob is~$N$ but we do not know~$N$. Let $\expect{\op{A}_y^{k}}$ be the expected value of Bob's output when the input is~$k$ and~$i$ is even. Similarly, $\expect{\op{A}_n^{k}}$ represents the expected value for input~$k$ and~$i$ odd. One way to verify whether Bob is quantum or classically memory bound is to initially query for a choice of~$k$ and then sequentially query for increasing size of~$k$. For each choice of~$k$ we choose~$i$ to be odd or even with a uniform random distribution. We define~$V$ to measure how successful Bob is in correctly identifying the string~$w$. Performing the procedure~$Q$ times gives us a value
\begin{equation}
V= \sum_{j=1}^Q \expect{\textsf{A}_y^{k=4j}} - \expect{\textsf{A}_n^{k=4j}},
\end{equation}
where we have chosen to increase the input~$k$ by multiples of~$4$ at each iteration. If~$k \leqslant \log N -1$, then for both the classical and quantum case Bob can achieve $V=Q.$ For~$k > \log N -1$, since there is no perfect strategy in the classical case we have~$V<Q$, while the quantum strategy still achieve~$V=Q$. The classical value can be made much tighter since the optimal classical strategy for~$k > \log N -1$ is just a random guess. We have shown in Section~\ref{sec:succinctness} that allowing error classically does not help in terms of reducing the memory requirement, i.e. bounded-error rtPFAs need at least $ 2^{k+1} $ states to solve $ \tt EVENODD^k $. This implies that the classical value is bounded by~$\frac{\log N -1}{4}$. Similar inequalities may be derived for $\pal$ and $\promisepal$ and they are summarized in Tables~\ref{tab:classical} and~\ref{tab:quantum}.
\vspace*{-1.5em}
\begin{table}[!ht]
\centering
\begin{tabular}{@{}cccc@{}}
Problem & Type  & Classical Memory & Inequality Value \\\hline \\ [-1.25ex]
$\pal$ & language recognition & $\log n$ &$\frac{2^N}{4}$  \\[1ex]
$\evenodd$ & family of promise problems & $2^{k+1}$ &$\frac{\log N -1}{4}$       \\[1ex]
$\promisepal$ & promise problem & $\log n$ &$\frac{2^N}{4}$       
\end{tabular}
\caption{For both $\pal$ and $\promisepal$ no 2PFA exists that solves the problem with bounded error. For the family of $\{\evenodd \; | \; k>0\}$, there is no bound on the number of states for
real-time PFAs that solve the members of this family with bounded error. Given an input string of size~$n$ and classical memory~$N$, the table gives the memory requirement for solving the specific instance and the value attained for the non-contextual inequality.}
\label{tab:classical}
\end{table}
\vspace*{-4.0em}
\begin{table}[!ht]
\centering
\begin{tabular}{@{}cccccc@{}}
Problem & Quantum Model && Quantum Memory && Quantum Value  \\\hline \\ [-1.25ex]
$\pal$ & 2QCFA && qubit &&$Q-\delta$  \\[1ex]
$\evenodd$ & Real-time && qubit &&$Q$       \\[1ex]
$\promisepal$ & 2QCFA && qubit &&$Q$       
\end{tabular}
\caption{The weakest known quantum models that solve the given problems and the associated error in the solution. The value attained for the inequality is related to the number of runs~$Q$ of the game.}
\label{tab:quantum}
\end{table}
\vspace*{-2.5em}

It may be possible to improve these inequalities by finding other problems for which we obtain a similar separation as $\promisepal$ but with an exponential classical memory requirement and a polynomial time quantum automata.

\section{Discussion}
\label{discuss}

Perhaps the most alluring charm of quantum automata separations is the possibility they offer of constructing a computational device, that could solve a problem which no classic device with finite memory could. Rather than just ruling out hidden variables with an exponential size increase, these computational devices could be used in principle to rule out arbitrary size hidden variables by increasing the problem input size. The thorny issue though is a trade-off between the memory utilized and the amount of precision required in the interactions with the quantum memory. This precision requirement appears either in the form of the matrix entries for the unitaries, as in the case of $\evenodd$ \emph{or} in form of the ability to resolve two states that may be arbitrarily close to each other, as in $\pal$. 

Any experimental setup is always restricted by some level of precision. If not by technological limitations then by more fundamental restrictions such as the uncertainty principle. The question that we are inevitably led to consider following this line of reasoning is whether it is possible to retain the quantum advantage in the automata model while still requiring finite precision in our interactions with the memory.

The intersection of ideas from classical simulation of contextuality and automata theory leads us to the notion of \emph{Finite Precision Quantum Automata} (FPQA). A FPQA satisfies in addition to the usual automata requirements the constraints that for any two unique unitaries~$\textsf{U}_i$ and~$\textsf{U}_j$ applied during the computation we have $\abs{\textsf{U}_i - \textsf{U}_j} > \epsilon$ and for any two different states~$\ket{\psi}$ and $\ket{\phi}$ obtained during the computation we have $\abs{\inner{\psi}{\phi}}^2 > \delta$. 

It is not clear whether we can construct a FPQA that still manages to provide a computational advantage over classical automata. Meyer~\cite{Meyer99} has argued that the Kochen-Specker theorem~\cite{Kochen67} does not hold when only finite precision measurements available. Clifton and Kent~\cite{Clifton00} have generalized the arguments of Meyer for POVM's. On ther other hand Mermin and Cabello~\cite{Cabello02} have indepedently argued that such nullification theorems do not hold. Recently Cabello and Cunha~\cite{Cabello11b} have proposed a two-qutrit contextuality test, claiming it to be free of the finite precision loophole. These tests though do admit a finite memory simulation model. Constructing a FPQA that yields a separation over the classical models would not admit a finite memory simulation model and consequently provide a stronger separation.

Even if the FPQA model does turn out to be equivalent to the classical model in terms of the class of problems it solves, it does not take away from the succintness advantage of the quantum model. In the previous section we argued that quantum automata separations serve as witnesses for distingishing between genuinely quantum and space bounded classical players. We can flip the reasoning around and observe that simulating quantum contextuality is an inherently classical memory intensive task. This difficulty can be used to constuct classical Proofs of Space as identified by Dziembowski et al.~\cite{Dziembowski13}. The idea is to establish that Bob has access to a certain amount of memory. Bob is asked to simulate an appropriately chosen quantum contextuality scenario. This would require exponential memory on Bob's side while the verifier could directly check that Bob's output satisfies the required quantum correlations. Note also that by definition, pre-computation does not allow Bob to simulate quantum contextuality. This task would require identification of generalizable quantum contextual promise problems. We note that use of the term `contextual' here is different than the standard notion of context-free languages.

\bibliographystyle{plain}
\bibliography{tcs}

\newpage

\appendix

\section{The proof of Theorem \ref{thm:nosubPTMforPAL}}
\label{app:nosubPTMforPAL}

We~give a proof of the fact that 2PFAs cannot solve promise problem $ \promisepal $ with bounded-error. We use the technique given in~\cite{DS92} that shows $ \pal $ cannot be recognized by bounded-error 2PFAs. We begin with the following fact.

\begin{fact}
	\label{fact:DS92}
	(Theorem 3.3 on Page 809 of \cite{DS92})
	\\ \\
	Let $ A,B \subseteq \Sigma^*$ with $ A \cap B = \emptyset$. Suppose there is an infinite set $I$ of positive integers and, for each $m \in I$, a set $W_m \subseteq \Sigma^*$ such that
	\begin{enumerate}
	\item $ |x| \leq m $ for all  $ x \in W_m$,
	\item for every integer $k$ there is an $m_k$ such that $|W_m| \geq m^k$ for all $ m \in I $ with $ m \geq m_k $, and
	\item for every $ m \in I $ and every $ x,x' \in W_m $ with $ x \neq x' $, there are strings  $ y,z \in \Sigma^* $ such that either $ yxz \in A $ and $ yx'z \in B $, or $ yxz \in B $ and $ yx'z \in A $.
	\end{enumerate}
	Then no bounded-error 2PFA separates $ A $ and $ B $.
\end{fact}

\begin{theorem}
	Bounded-error 2PFAs cannot solve $ \promisepal $.
\end{theorem}
\begin{proof}
	Let $ m > 0 $ and $ \{u_1,\ldots,u_{2^m} \} $ be all strings in $ \{a,b\}^m $. We define $ W_m $ as follows:
	\begin{equation*}
		W_m = \{ u_i c u_i \mid 1 \leq i \leq 2^m  \}. 
	\end{equation*}
Since the size of $W_m$ is super-polynomial in $m$, i.e., $ |W_m| = 2^m $ we satisfy the second condition given in Fact~\ref{fact:DS92}.
	
Let $ x = u_{i} c u_{i} $ and $ x' = u_{j} c u_{j} $ be two different elements of $ W_m $ ($ i \neq j $). We pick $ y $ as $ u_i^r $ and $ z $ as $ u_j^r $. Then, $ yxz = u_i^r u_i c u_i u_j^r \in \promisepalyes $ and $ yx'z = u_i^r u_j c u_j u_j^r \in \promisepalno $. Thus, no bounded-error 2PFA can solve $ \promisepal $ due to Fact \ref{fact:DS92}.
	\qed
\end{proof}

Freivalds and Karpinski \cite{FK94} presented a slightly different fact to prove that $ \pal $ cannot be recognized by any bounded-error sublogarithmic probabilistic Turing machine (PTM). 

\section{The proof of Theorem \ref{thm:TWIN}}
\label{app:TWIN}

Instead of calling the 2QCFA (for $\tt PAL$) of Ambainis and Watrous \cite{AW02} as a whole, we call a subroutine of this algorithm and use it in our new algorithm. We~refer to it as $ \awpal $ and its details are as follows. $ \awpal $ uses a quantum register with 3 states, i.e. $ \ket{q_1}, \ket{q_2}, \ket{q_3} $, which is set to $ \ket{q_1} $ at the beginning. $ \awpal $ reads the given input, say  $ w \in \{a,b\}^* $, from left to right twice. In the first pass, it applies $ \op{U}_{\sigma} $ to the quantum register for each $ \sigma \in \{a,b\} $, and, in the second pass, it applies $ \op{U}_{\sigma}^{-1} $ to the quantum register for each $ \sigma \in \{a,b\} $, where
\begin{equation*}
	\op{U}_a = \dfrac{1}{5} \left(  
	\begin{array}{rrr}
		4 & 3 & 0
		\\
		-3 & 4 & 0
		\\
		0 & 0 & 5
	\end{array}	
	 \right)
	 \mbox{ and }
	\op{U}_b = \dfrac{1}{5} \left(  
	\begin{array}{rrr}
		4 & 0 & 3
		\\
		0 & 5 & 0
		\\
		-3 & 0 & 4
	\end{array}	
	 \right).
\end{equation*}
The quantum state of $\awpal $ after the two passes is given by:
\begin{equation*}
	\ket{\psi_{f}} = \op{U}_{|w|}^{-1} \op{U}_{|w|-1}^{-1} \cdots \op{U}_{w_2}^{-1} \op{U}_{w_1}^{-1} \op{U}_{|w|} \op{U}_{|w|-1} \cdots \op{U}_{w_2} \op{U}_{w_1} (1~0~0)^T
\end{equation*}
The two passes are followed by a measurement in the computational basis. From \cite{AW02}, we know that 
\begin{itemize}
	\item If $ w \in \pal $, then $ \ket{\psi_f} = \ket{q_1} $, and so, the measurement result is the first state.
	\item If $ w \notin \pal $, then $ \ket{\psi_f} \neq \ket{q_1} $ and the second or third state is observed with a probability at least $ 25^{-|w|} $.
\end{itemize}

Now, we describe an exact rational restarting rtQCFA for $ \exacttwinpal $, a modified version of $\promisepal$. i.e.
\begin{equation*}
	\promisetwinpal = (\promisetwinpalyes,\promisetwinpalno), \mbox{ where}
\end{equation*}
\begin{itemize}
	\item $ \promisetwinpalyes = \{ u c u c v c v | u,v \in \{a,b\}^+, |u|=|v|, u \in \pal, \mbox{ and } v \notin \pal \} $, and
	\item $ \promisetwinpalno = \{ u c u c v c v | u,v \in \{a,b\}^+, |u|=|v|, u \notin \pal, \mbox{ and } v \in \pal \} $.
\end{itemize}

Let $ w = ucucvcv $, a member of $ \promisetwinpal  $, be the input such that $ u,v \in \{a,b\}^* $ and $ |u|=|v| $. $ \exacttwinpal $ reads the input only from left to right in an infinite loop. A single iteration of this loop proceeds as follows. $ \exacttwinpal $ starts with quantum state $ (1 ~~ 0 ~~ 0)^T $. Then, $ \exacttwinpal $ splits the computation into two branches, $ branch_1 $ and $ branch_2 $. For this purpose, $ \exacttwinpal $ applies $ \op{U}_a $ to the quantum register and then measures it in the computational basis. $ \exacttwinpal $ continues with $ branch_1 $ (resp., $branch_2$) with probability $ \frac{16}{25} $ (resp., $\frac{9}{25}$) when the first (resp., the second) state is observed. In $branch_1$,  $ \exacttwinpal $ executes $ \awpal $ for $v$ on $vcv$. Similarly, in $branch_2$, $ \awpal $ is executed for $u$ on $ucu$. Here $\awpal$ reads $ v $ (or $u$) twice by scanning $ vcv $ (or $ucu$). $ \exacttwinpal $ gives its decision based on the outcome of $\awpal$ at the end.  If the second or third state is observed in $ branch_1 $ (resp., $branch_2$), then $w$ is accepted (resp., rejected) by $ \exacttwinpal $. Otherwise, the computation continues with the next iteration. The analysis of a single iteration is as follows:
\begin{itemize}
	\item Suppose that $ w \in \promisetwinpalyes $. Then, $ branch_1 $ ends in a state different than~$\ket{q_1}$ with probability at least $ 25^{-|v|} $ since $ v \notin \pal $, and, $ branch_2 $ ends in state~$\ket{q_1}$ always since $u \in \pal$. Therefore, a decision is given only in the first branch. That is, $w$ is accepted with a probability at least $ 16.25^{-|v|-1} $.
	\item Suppose that $ w \in \promisetwinpalno $. Then, $ branch_1 $ ends in state~$\ket{q_1}$ always since $u \in \pal$, and, $ branch_2 $ ends in a state different than~$\ket{q_1}$ with probability at least $ 25^{-|u|} $ since $ u \notin \pal $. Therefore, a decision is given only in the second branch. That is, $w$ is rejected with probability at least $ 9.25^{-|u|-1} $.
\end{itemize}
Therefore, the members of $\promisetwinpalyes$ (resp., $\promisetwinpalno$) are accepted (resp., rejected) exactly in expected exponential time as stated in Theorem~\ref{thm:TWIN}.

\section{The proof of Theorem \ref{thm:nosubPTMforTWIN}}
\label{app:nosubPTMforTWIN}

In order to prove that 2PFAs (or sublogarithmic space PTMs) cannot solve promise problem $ \promisetwinpal $ with bounded-error, we use a reduction technique used in \cite{YFSA12A}. It was shown that $ \twin =\{ucu \mid u \in\{a,b\}^* \} $ cannot be recognized by bounded-error 2PFAs. The proof is based on the following contradiction: Assume $ \twin $ can be recognized by a bounded-error 2PFA. Then $ \pal $ can also be recognized  by a bounded-error 2PFA. Since $ \pal $ is known to be not recognizable by bounded-error 2PFA~\cite{DS92}, we end up with a contradiction and our assumption must have been incorrect.

Suppose that a bounded-error 2PFA, say $ \mathcal{P} $, can solve $ \promisetwinpal $. Then we can construct a bounded-error 2PFA, say $\mathcal{P'}$, that solves $ \promisepal $ as follows. For a given input $ w = ucv $, a member of $ \promisepal $, $ \mathcal{P'} $ can simulate $ \mathcal{P}  $ on $ w' = ucucvcv $ and gives the same output as $ \mathcal{P}$. The simulation is straightforward: After leaving a $ c $ or an end-marker, the head of $ \mathcal{P} $ is on the first or the last symbol of  $ u $ or $ v $. $ \mathcal{P'} $ can easily keep a track of such leaves by its internal states, and so, it can place its head at the corresponding place of $ u $ or $ v $, respectively. So, $ \mathcal{P'} $ can solve $\promisepal$ with bounded error. But, as shown previously, this is impossible. We can derive the same result also for bounded-error sub-logarithmic space PTMs.

\section{Exact 2PFAs}
\label{app:exact2PFAs}

From an exact 2PFA, we can obtain two two-way nondeterministic finite automata, say $ \mathcal{N}_1 $ and $ \mathcal{N}_2 $, that give the decision of ``acceptance'' for yes instances and no instances, respectively. So, there are two rtDFAs $ \mathcal{D}_1 $ and $ \mathcal{D}_2 $ with the same behaviour. Then, we can obtain the desired rtDFA by taking the tensor product of $ \mathcal{D}_1 $ and $ \mathcal{D}_2 $. For each member of the promise problem, either $ \mathcal{D}_1 $ or $ \mathcal{D}_2 $ gives the decision of ``acceptance''.

\section{The proof of Theorem \ref{thm:promiseexptwinpal}}
\label{app:promiseexptwinpal}

We start by analyzing the long term behaviour of a single branch of $ \exacttwinpal $ (see Appendix \ref{app:TWIN}). Let $ w=ucucvcv $ be a member of $\promisetwinpalyes$ such that $ u,v \in \{a,b\}^*$ and $ |u|=|v|=n $. Suppose that we execute $branch_1$ $ T=25^{n}  $ times on $w$. The overall accepting probability is given by
%\begin{eqnarray*}
%	\sum_{i=1}^{T} 25^{-n} (1-25^{-n})^{i-1} & = & 25^{-n} \left( \dfrac{1 - (1-25^{-n})^{T}}{1-(1-25^{-n})} \right) 
%	\\
%	& = & 1 - \left(1-\frac{1}{25^n}\right)^{25^{n}}
%	\\
%	& = & 1 - \left(1-\frac{1}{25^n}\right)^{25^{n}} .
%\end{eqnarray*}
\begin{equation*}
	\sum_{i=1}^{T} 25^{-n} (1-25^{-n})^{i-1}  =  25^{-n} \left( \dfrac{1 - (1-25^{-n})^{T}}{1-(1-25^{-n})} \right) 
	 = 1 - \left(1-\frac{1}{25^n}\right)^{25^{n}}.
\end{equation*}
Since $ \lim_{x \rightarrow \infty} \left( 1-\frac{1}{x} \right)^x = \frac{1}{e} $ and $ \left( 1-\frac{1}{x} \right)^x $ is less than $ \frac{1}{e} $ for every $ x \in \mathbb{Z}^{+} $, the overall accepting probability is always greater than
\begin{equation*}
	1 - \dfrac{1}{e}.
\end{equation*}
The same result can also be obtained for the overall rejecting probability of $ branch_2 $ when executed on a member of $\promisepalno$, say $ w=ucucvcv $, $ 25^{n} $ times, where $ u,v \in \{a,b\}^* $ and $ |u|=|v|=n $.

Remember the definition of our promise problem:
\begin{equation*}
	\promiseexptwinpal = (\promiseexptwinpalyes,\promiseexptwinpalno),
\end{equation*}
where 
\begin{itemize}
	\item $ \promiseexptwinpalyes = \{ (u c u c v c v c)^{t} | u,v \in \{a,b\}^+, |u|=|v|, u \in \pal, v \notin \pal, \mbox{ and } t \geq 25^{|u|}  \} $, and
	\item $ \promiseexptwinpalno = \{ (u c u c v c v c)^{t} | u,v \in \{a,b\}^+, |u|=|v|, u \notin \pal, v \in \pal, \mbox{ and } t \geq 25^{|u|} \} $.
\end{itemize}
We now give a Las Vegas rational rtQCFA $ \lvexppal $ for the problem. Let $ w = (ucucvcvc)^{t} \in \promiseexptwinpal $, where $ u,v \in \{a,b\}^* $, $ |u|=|v|>0 $, and $ t \geq 25^{|u|} $.  At the beginning, the computation is split into two branches with probabilities $ \frac{16}{25} $ and $ \frac{9}{25} $ (as described before). In the first (resp., the second) branch, $ branch_1 $ (resp., $ branch_2 $) of $ \exacttwinpal $ is implemented on each block of $ ucucvcvc $ until the end of input. $ \lvexppal $ gives the decision of ``don't know'' if the input has not been accepted or rejected. Using our analysis earlier in the section, we can conclude that the probabilities of decisions given by $ \lvexppal $ are given by:
\begin{itemize}
	\item Any member of $ \promiseexptwinpalyes $ is accepted with probability at least $ \frac{16}{25} \left( 1 - \frac{1}{e} \right) $. So, $ \lvexppal $ gives the decision ``don't know'' with a probability no more than $ \frac{9}{25} + \frac{16}{25e} $.
	\item Any member of $ \promiseexptwinpalno $ is rejected with probability at least $ \frac{9}{25} \left( 1 - \frac{1}{e} \right) $. So, $ \lvexppal $ gives the decision ``don't know'' with a probability no more than $ \frac{16}{25} + \frac{9}{25e} $.
\end{itemize}
Executing many copies of $ \lvexppal $ in parallel, we can increase both the accepting and rejecting probabilities arbitrary close to 1. The probability of ``don't know'' then gets closer to 0. To eliminate ``don't answer'', we can restart $ \lvexppal $ whenever we obtain a ``don't know'' answer. Thus, we can obtain a linear expected time rational exact restarting rtQCFA solving $ \promiseexptwinpal $.

\section{The proof of Theorem \ref{thm:nortPFA}}
\label{app:nortPFA}

We show that no bounded-error rtPFAs can solve $ \promiseexptwinpal $. Suppose that there exits such a rtPFA called $ \mathcal{P} $ that works with error probability $ \frac{1}{3} $. Then, by using $ \mathcal{P} $, we can construct a bounded-error 2PFA, say $ \mathcal{P'} $, for $ \promisetwinpal $. Let $ w = ucucvcv $ be the input such that $ u,v \in \{a,b\}^+ $ and $|u|=|v|=n>0$. $ \mathcal{P'} $ executes $ \mathcal{P} $ on $ (ucucvcv)^i $ for each $ i>0 $ and then follows the decision of $ \mathcal{P} $ with a certain probability. More specifically, $ \mathcal{P'} $ can easily  execute $ \mathcal{P} $ on $ (ucucvcv)^i $ for each $ i>0 $ in an infinite loop. That is, after the $ i^{th} $ iteration, $ \mathcal{P} $ is fed $ (ucucvcv)^i $. At this point, before continuing with the next iteration, the decision of $ \mathcal{P} $ on $ (ucucvcv)^i $ is implemented with probability $ 25^{-16n} $. This is performed by reading $ O(n) $ symbols. With the remaining probability, $ 1-25^{-16n} $, $ \mathcal{P'} $ continue with the next iteration, i.e. it feeds another $ ucucvcvc $ to $ \mathcal{P} $. Note that the last state of $ \mathcal{P} $ at the end of the $ i^{th} $ iteration is recorded by the internal states of $ \mathcal{P'} $ and $ \mathcal{P} $ is set to this state again just before the next iteration starts.

Let $ T=25^n $. Assume that $ w \in \promisetwinpalyes $. Then, after the $ (25^{n})^{th} $ iteration, we know that $ \mathcal{P} $ gives the decision of acceptance with probability at least $ \frac{2}{3} $. In the worst case, we can assume that $ \mathcal{P} $ gives the rejection decision before the $ (25^{n})^{th} $ iteration. Then the overall accepting probability of $ \mathcal{P'} $ can be calculated as follows.

The input is rejected with probability 
\begin{equation*}
	\sum_{i=1}^{T-1} 25^{-16n} (1 - 25^{-16n})^{i-1} < 1 - \sqrt[16]{\frac{1}{e}}
\end{equation*}
before the $ T^{th} $ iteration. So, the computation continues from the  $ T^{th} $ iteration with probability at least $ \sqrt[16]{\frac{1}{e}} $, which is greater than $ 0.9 $. After this point, the input is accepted with a probability at least $ \frac{2}{3} $ and so the overall accepting probability is always greater than $ 0.6 $. Similar analysis holds for the members of $ \promisetwinpalno $, and so, the rejecting probability of such members also exceeds $ 0.6 $. Thus, $ \mathcal{P'} $ can solve $ \promiseexptwinpal $ with bounded error. This is a contradiction. 

\section{The proof of Theorem \ref{thm:polynomial}}
\label{app:polynomial}

2QCFAs can recognize $ \eq $ in polynomial expected time \cite{AW02} with one-sided bounded-error. The given algorithm executes two phases, a quantum and then a classical one, in an infinite loop. We will use only the first phase and call it $ \aweq $ which operates on a single qubit ($ \ket{q_1}, \ket{q_2} $) and uses only algebraic transitions. The input is read by $ \aweq $ in realtime mode, and if it is $ a^mb^m $, then $ \ket{q_1} $ is observed exactly, and if it is $ a^mb^n $ ($ m \neq n $), then $ \ket{q_2} $ is observed with a probability at least $ \frac{1}{2(m-n)^2} $. Thus, similar to the previous algorithm, we can easily obtain a quadratic expected time exact algebraic restarting rtQCFA algorithm solving $ \promiseeq $.

Recently Yakary{\i}lmaz \cite{Yak12C} presented a linear-time one-sided bounded-error one-way QFA (1QFA) algorithm for the language $ \{a^nba^n\} $ where the model uses \textit{quantum head}. This algorithm can be easily modified to use only rational numbers (see \cite{Wat97,Wat98}). Then, we can follow that $ \promiseeq $ can be solved by a Las Vegas rational 1QFA in linear time, and so by an exact rational two-way QFAs in linear expected time (see \cite{YS10B}).

\section{The proof of Theorem \ref{thm:promiseeq}}
\label{app:promiseeq}

In \cite{DS89,DS90}, Dwork and Stockmeyer showed that 2PFAs or $ o(\log\log n) $-space PTMs can recognize a non-regular language only in super-polynomial expected time. Their proofs follows from some technical facts. We will review the related facts and then modify them for our aim.

For a given language $ \mathtt{L} \subseteq \Sigma^* $ and a length $ n $, two strings whose lengths are no more than $ n $, say $ w $ and $w'$, are called $ n $-dissimilar, i.e. $ w \nsim_{\mathtt{L},n} w' $, if there exists a string $ u \in \Sigma^* $ satisfying $ | wu | \leq n $ and $ | w'u | $ such that 
\[
	wu \in \mathtt{L} \mbox{ if and only if } w'u \in \overline{\mathtt{L}}.
\] 
We call $ u $ as a separator of $ w $ and $ w' $. $ N_\mathtt{L}(n) $ is the maximum number of pairwise $ n $-dissimilar strings. 
\begin{fact} (Lemma 4.3 on Page 1017 in \cite{DS90}) 
	For every $ \epsilon < \frac{1}{2} $, there are positive constants $ \alpha_\epsilon $ and $ \eta_\epsilon $ such that, if a 2PFA $ \mathcal{M} $ having $ c $ states recognizes the language $ \mathtt{L} $ with error bound $ \epsilon $ and within expected time $ T(n) $, then
\begin{equation}
	\label{eq:DS90}
	\left( \alpha_\epsilon c \left( \log T(n) + \log cn \right) \right)^{n^2} \geq N_\mathtt{L}(n) \mbox{ for all  } n \geq  \eta_\epsilon.
\end{equation}
\end{fact}
The proof was given based on a contradiction such that if Equation \ref{eq:DS90} does not hold then two $ n $-dissimilar strings, say $ w $ and $ w' $, cannot be distinguished by $ \mathcal{M} $. That is, if $ wu $ is accepted by $ \mathcal{M} $ with a probability at least $ 1 - \epsilon $, then $ w'u $ is accepted by $ \mathcal{M} $ with a probability at least $ \frac{1}{2} $, where $ u $ is a separator of $ w $ and $ w' $. The good news is that the proof still works if we focus on promise problems after modifying the dissimilarity relation as follows.

For a given promise problem $ \mathtt{P} = (\mathtt{P_{yes}},\mathtt{P_{no}}) $ defined on $ \Sigma $ and a length $ n $, two strings  whose lengths are no more than $ n $, say $ w $ and $w'$, are called $ n $-dissimilar, i.e. $ w \nsim_{\mathtt{P},n} w' $, if there exists a string $ u \in \Sigma^* $ satisfying $ | wu | \leq n $ and $ | w'u | $ such that 
\[
	wu \in \mathtt{P_{yes}} \mbox{ if and only if } w'u \in \overline{\mathtt{P_{no}}}.
\] 
$ N_\mathtt{P}(n) $ is the maximum number of pairwise $ n $-dissimilar strings. Now, we can rephrase the fact above for promise problems.
\begin{theorem}
	For every $ \epsilon < \frac{1}{2} $, there are positive constants $ \alpha_\epsilon $ and $ \eta_\epsilon $ such that, if a 2PFA $ \mathcal{M} $ having $ c $ states solves the promise problem $ \mathtt{P} $ with error bound $ \epsilon $ and within expected time $ T(n) $, then
\begin{equation}
	\label{eq:DS90-2}
	\left( \alpha_\epsilon c \left( \log T(n) + \log cn \right) \right)^{n^2} \geq N_\mathtt{P}(n) \mbox{ for all  } n \geq  \eta_\epsilon.
\end{equation}
\end{theorem}
It is clear that if $ N_\mathtt{L}(n) $ (or $ N_\mathtt{P}(n) $) is not less that $ n^{\delta} $ for some $ \delta > 0 $ for infinitely many $ n $, then the corresponding bounded-error 2PFA cannot have a sub-exponential expected time. As described in \cite{DS90}, if we augment a 2PFA with a $ s(n) $-space work tape, a $ s(n) $-space PTM, then we can safely replace $ c $ with $ c(n) = 2^{ds(n)} $ in Equations \ref{eq:DS90} and \ref{eq:DS90-2}, where $ d $ is a constant depending on the machine. Then, after simple calculations, we can follow our sub-exponential expected time statement also for bounded-error $ o(\log \log n) $-space PTMs.

Now, we show that $ N_{\promiseeq}(n) \geq \Omega(n) $ which is sufficient to conclude that $ \promiseeq $ cannot be solved by any bounded-error $ o(\log \log n) $-space PTMs in polynomial expected time. We define a set of $ m $ strings $ \{ a^i \mid 1 \leq i \leq m \} $. These $m$ strings are $ (3m+1) $-dissimilar since for any two different $ w=a^i $ and $ w'=a^j $, there exists $ u = ba^iba^j $ such that
\[
	wu = a^iba^iba^j \in \promiseeqyes \mbox{ and } w'u = a^jba^iba^j \in \promiseeqno.
\]

\section{The proof of Theorem \ref{thm:succintness}}
\label{app:succintness}

The behaviour of a rtPFAs over a single letter can be modelled as a Markov chain \cite{AF98,MP01,MPP01}. We utilize a fact given in~\cite{MP01}.

\begin{fact}
	(Theorem 6 on Page 486 of \cite{DS92})
	For any $ m \times m $ stochastic matrix $ A $, there is an integer $ d $ such that the limit
	\begin{equation*}
		\lim_{i \rightarrow \infty} \left( A^d \right)^i
	\end{equation*}
	exists. Furthermore, there are $ k \geq 1 $ positive integers $ d_1,\ldots,d_k $, which depend on $ A $ such that $ d_1 + \cdots + d_k \leq m $ and $ d = lcm(d_1,\ldots,d_k) $.
\end{fact}

Let $ \mathcal{P} $ be an $ m $-state rtPFA defined over a unary alphabet $ \Sigma=\{a\} $ and $ A $ be its stochastic transition matrix defined for $ a $. Based on the above fact, we can conclude that the series defined by the accepting probabilities of the strings \begin{equation*} \{ a^{d},a^{2d},a^{3d},\ldots \} \end{equation*} has a limit, say $ a_1 $. Similarly, that of \begin{equation*} \{ a,a^{d+1},a^{2d+1},a^{3d+1},\ldots \} \end{equation*} has also a limit, say $ a_2 $. In a similar fashion, we can define (at most) $ d $ limiting accepting probabilities: $ \{a_1,\ldots,a_{d}\} $. Let $ \delta>0 $ be sufficiently small such that $ (a_i-\delta,a_i+\delta) $ and $ (a_j-\delta,a_j+\delta) $ do not overlap unless $ a_i \neq a_j $, where $ 1 \leq i < j \leq d $. We name each interval $ (a_i-\delta,a_i+\delta) $ as a region called $ r_i $, where $ 1 \leq i \leq d $. Let $ M_{\delta} $ be a sufficiently big integer such that the accepting probability of each string longer than $ M_{\delta} $ lies in one of the regions defined above and $ a^{M_{\delta}} $ lies in $ a_1 $. Thus, the accepting probabilities of strings
\begin{equation*} 	a^{M_{\delta}},a^{M_{\delta}+1},a^{M_{\delta}+2},\ldots,a^{M_{\delta}+d-1},a^{M_{\delta}+d},a^{M_{\delta}+d+1},\ldots
\end{equation*}
lie in regions
\begin{equation*}
	r_1,r_2,r_3,\ldots,r_d,r_1,r_2,\ldots,
\end{equation*}
respectively. We now observe that the regions are visited in a cycle. Any unary deterministic finite automaton can enter a similar cycle for long strings as well. Ambainis and Yakary{\i}lmaz showed that \cite{AY12} any rtDFA can solve $ \tt EVENODD^k $ only if $ gcd(d,2^k) \neq gcd(d,2^{k+1}) $. That is, the length of the cycle must be a multiple of $ 2^{k+1} $. The same argument is also true for our case. The details are given below.

Suppose that $ \mathcal{P} $ can solve $ \tt EVENODD^k $ with an error bounded $ \epsilon \in (0,\frac{1}{2}) $ and $ m < 2^{k+1}  $. Without lose of generality, we can also assume that if $ a_i $ is in $ (\epsilon,1-\epsilon) $, then $ r_i $ does not contain both $ \epsilon $ and $ 1-\epsilon $, where $ 1 \leq i \leq d $. Therefore, it is clear that the members of $ \tt EVENODD^k_{yes} $ and $ \tt EVENODD^k_{no} $ longer than $ N_\delta $ must lie in the different regions. Let $ L=2^l=gcd(d,2^{k+1}) $ for some non-negative integer $ l $. Since $ m < 2^{k+1} $, $ l $ can be at most $ k $. (Remember that $ d = lcm(d_1,\ldots,d_k) $ and each $d_i$ is less than $m$, where $1 \leq i \leq k$ for some $k \geq 1$. Therefore, each $d_i$ can contain at most $l$ $2$(s) as its factor(s).) Therefore, $ gcd(d,2^k) $ is equal to $ 2^l $ as well. The accepting probabilities of strings
\begin{equation*}
	\{a^{i2^k} \mid 2^{k} \geq N_\epsilon \mbox{ and } i > 0 \}
\end{equation*}
lie in $ t=\frac{d}{2^l} $ different $ r $-regions, i.e. $ R = \{r_1',\ldots,r_t'\} $. More precisely, 
\begin{itemize}
	\item the accepting probability of $ a^{1 \cdot 2^k} $ lies in $ r_1' $,
	\item the accepting probability of $ a^{2 \cdot 2^k} $ lies in $ r_2' $,
	\item the accepting probability of $ a^{3 \cdot 2^k} $ lies in $ r_3' $,
	\item ...
	\item the accepting probability of $ a^{t \cdot 2^k} $ lies in $ r_t' $,
	\item the accepting probability of $ a^{(t+1) \cdot 2^k} $ lies in $ r_1' $,
	\item ... .
\end{itemize}
We can divide these regions into two classes, namely \textit{yes regions} and \textit{no regions}:  yes regions are closer to $ 1-\epsilon $ and no regions are closer to $ \epsilon $. Here a yes region corresponds to a member of $ \tt EVENODD_{yes}^k $ and a no region corresponds to a member of $ \tt EVENODD_{no}^k $. Therefore, each class must not be empty. Now, we focus on only the members of $ \tt EVENODD_{no}^k $: The accepting probabilities of strings
\begin{equation*}
	\{a^{i2^k} \mid 2^{k} \geq N_\epsilon \mbox{ and } i \mbox{ is an odd positive integer} \}
\end{equation*}
lie in also $ t=\frac{d}{2^l} $ different $ r $-regions that are actually a subset of $ R $. More precisely,
\begin{itemize}
	\item the accepting probability of $ a^{1 \cdot 2^k} $ lies in $ r_1' $,
	\item the accepting probability of $ a^{3 \cdot 2^k} $ lies in $ r_3' $,
	\item the accepting probability of $ a^{5 \cdot 2^k} $ lies in $ r_5' $,
	\item ...
	\item the accepting probability of $ a^{t \cdot 2^k} $ lies in $ r_t' $,
	\item the accepting probability of $ a^{(t+2) \cdot 2^k} $ lies in $ r_2' $,
	\item the accepting probability of $ a^{(t+4) \cdot 2^k} $ lies in $ r_4' $,
	\item ...
	\item the accepting probability of $ a^{(t+t-1) \cdot 2^k} $ lies in $ r_{t-1}' $,
	\item the accepting probability of $ a^{(t+t+1) \cdot 2^k} $ lies in $ r_1' $,
	\item ... .
\end{itemize}
Thus, yes regions must be empty. This is a contradiction. This is why $ gcd(d,2^k) $ must be different than $ gcd(d,2^{k+1}) $. It can be only possible when $ d $ is a multiple of $ 2^{k+1} $. As pointed in \cite{AY12}, it is straightforward that $ \tt EVENODD^k $ can be solvable by a $ 2^{k+1} $-state realtime deterministic finite automaton. 

\end{document}